\documentclass[11pt,fleqn]{article}

\usepackage{amssymb,amsmath}
\usepackage{natbib}
\usepackage{vmargin,setspace,verbatim}
\usepackage{graphicx,graphics,color,tikz,tikz-network}
\usepackage{microtype}

\usepackage{libertine}
\usepackage[libertine]{newtxmath}

\definecolor{darkgreen}{rgb}{0,0.2,0}
\definecolor{darkred}{rgb}{0.3,0,0}

\usepackage[colorlinks=true, pdfstartview=FitV, linkcolor=darkgreen, citecolor=darkred, urlcolor=black]{hyperref}

\newcounter{llst}
\newenvironment{abet}{\begin{list}{\rm (\alph{llst})}{\usecounter{llst}
\setlength{\itemindent}{0em} \setlength{\leftmargin}{3em}
\setlength{\labelwidth}{2em} \setlength{\labelsep}{1em}}}{\end{list}}

\newcounter{llist}
\newenvironment{numm}{\begin{list}{\rm (\roman{llist})}{\usecounter{llist}
\setlength{\itemindent}{0em} \setlength{\leftmargin}{3.5em}
\setlength{\labelwidth}{2.5em} \setlength{\labelsep}{1em}}}{\end{list}}

\newtheorem{theorem}{Theorem}[section]

\newtheorem{definition}[theorem]{Definition}
\newtheorem{expl}[theorem]{Example}

\newtheorem{proposition}[theorem]{Proposition}

\newtheorem{dscrpt}[theorem]{Description}

\newcounter{axiomatiser}

\newcounter{myclaimcount}

\newenvironment{proof}[1][Proof]{\noindent \textbf{#1.} }{\hfill
\rule{0.5em}{0.5em}}

\newenvironment{example}{\begin{expl} \rm}{\hfill $\blacklozenge$
\end{expl}}

\setpapersize{A4}
\setmarginsrb{80pt}{40pt}{80pt}{60pt}{12pt}{10pt}{12pt}{30pt}

\begin{document}

\begin{titlepage}

\title{\textbf{Game theoretic foundations of \\ the Gately power measure for directed networks}\thanks{We thank participants in the BiNoMa 2023 workshop for their helpful comments on a previous draft of this paper. This paper also benefited from discussions with Ren\'e van den Brink for which we are very grateful.}}

\author{Robert P.~Gilles\thanks{Department of Economics, The Queen's University of Belfast, Riddel Hall, 185 Stranmillis Road, Belfast, BT9~5EE, UK. \textsf{Email: r.gilles@qub.ac.uk}} \and Lina Mallozzi\thanks{\textbf{Corresponding author:} Department of Mathematics and Applications, University of Naples Federico II, Via Claudio 21, 80125 Naples, Italy. E-mail: \textsf{mallozzi@unina.it}} }

\date{August 2023}

\maketitle

\begin{abstract}
\singlespace
\noindent
We introduce a new network centrality measure founded on the Gately value for cooperative games with transferable utilities. A directed network is interpreted as representing control or authority relations between players---constituting a \emph{hierarchical} network. The power distribution of a hierarchical network can be represented through a TU-game. We investigate the properties of this TU-representation and investigate the Gately value of the TU-representation resulting in the Gately power measure. We establish when the Gately measure is a Core power gauge, investigate the relationship of the Gately with the $\beta$-measure, and construct an axiomatisation of the Gately measure.
\end{abstract}

\begin{description}
\singlespace
\item[Keywords:] Cooperative game; authority network; network centrality measure; Gately measure; axiomatisation.

\item[JEL classification:] A14; C71; D20
\end{description}

\thispagestyle{empty}

\end{titlepage}



\section{Introduction}

The concept of \emph{network centrality} has emerged from sociology, social network analysis and network science \citep{Newman2010book,Barabasi2016} into the field of cooperative game theory, giving rise to game theoretic methods to measure the most important and dominant nodes in a hierarchical social network \citep{Hierarchy2022}. The underlying method is to construct a cooperative game theoretic representation of characterising features of a network and to apply cooperative game theoretic analysis to create centrality measures for these networks.\footnote{We refer to \citet{Centrality2018} for an overview of the literature on cooperative game theoretic constructions of centrality measures in directed networks.}

We limit ourselves to directed networks as representations of collections of hierarchical or control relationships between the constituting players in a network. Such hierarchical relationships can be found in employment dynamics between managers and subordinates, the interaction between a professor and her students, rivalry between different sports teams based on past performance between them, or the connections between a government agent and the individuals they oversee. We refer to these relationships as ``hierarchical'' since it implies that the predecessor node has some level of control or authority over the successor node.\footnote{This authority can also be psychological and be influential on reputational features in the relationship. An example might be the relationship between two chess players. One of these players can have a psychological advantage over the other based on outcomes of past games between the two players and/or the Elo rating differential between the two players. } Taking this interpretation central, we refer to these directed networks as \emph{hierarchical}. 

A hierarchical relationship is between a \emph{predecessor} and a \emph{successor}, where the predecessor exercises some form of control or authority over the successor. The most natural representation is that through a TU-game that assigns to every group of players their ``total number of successors'', which can be interpreted in various ways. We consider the two standard ways: Simply counting the successors of all members of the group, i.e., the number of players that have at least one predecessor that is member of the group; or, counting the number of players for which \emph{all} predecessors are member of that group.\footnote{As \citet{Centrality2018} point out, the successor representations are only one type of representation of a hierarchical network. More advanced TU-representations have been pursued by \citet{Brink2023Network}.} We show that these two TU-representations are dual games. We remark that \citet{BrinkBorm2002} already characterised the main ``strong'' successor representation as a convex game \citep{Shapley1971}, implying that its dual ``weak'' successor representation is a concave game. 

A \emph{power gauge} for a network is now introduced as a vector of weights that are assigned to players in the hierarchical network that represent or measure each player's authority in that network. A \emph{power measure} is now introduced as a map that assigns to every hierarchical network a single power gauge. In this paper we investigate some power measures that assign such gauges founded on game theoretic principles related to the two TU-representations of hierarchical networks considered here. In particular, each simple hierarchical network---in which each player has at most one predecessor---has a natural power gauge in the form of the outdegree of each node in the network, representing the number of successors of a player in that network.

Application of the Shapley value \citep{Shapley1953} to the successor representations results in the $\beta$-power measure \citep{BrinkGilles1994}. This is the centre of the set of Core power gauges for each network. In the $\beta$-measure the weight of a player is equally divided among its predecessors. As such it has a purely individualistic foundation to measuring power. 

Applying the Gately value of the successor representations results in a fundamentally different conception of a power measure. Here, the set of dominated nodes is treated as a \emph{collective resource} that is distributed according to a chosen principle. In the Gately measure this is the proportional distribution.\footnote{We remark that other distribution principles can also be applied such as distributions founded on egalitarian fairness considerations. This falls outside the scope of the present paper.} This stands in contrast to the individualistic perspective of the $\beta$-measure. 

Since, the Gately measure is founded on such different principles, it is not a surprise that the assigned Gately power gauges are not necessarily Core power gauges. We identify conditions under which the assigned Gately power gauge is a Core power gauge in Theorem \ref{thm:GatelyCore}. In particular, we show that for the class of (weakly) regular hierarchical networks, the Gately power measure assigns a Core power gauge for that network. 

We are able to devise an axiomatic characterisation of the Gately value as the unique power measure that satisfies three properties. First, it is normalised to the number of nodes that have predecessors, which is satisfied by many other power measures as well. Second, it satisfies ``normality'' which imposes that a power measure assigns the full weight of controlling successors with no other predecessors and the power measure of the reduced network with only those nodes that have multiple predecessors. Finally, it satisfies a proportionality property in the sense that the power measure assigned is proportional to how many successors a node has. 

Finally, we address the question under which conditions the $\beta$- and Gately power measures are equivalent. We show that for the class of weakly regular hierarchical networks this equivalence holds. This is exactly the class of networks for which the Gately measure assigns a Core power gauge. This insight cannot be reversed, since there are non-regular networks for which the Gately and $\beta$-measures are equivalent. 

\paragraph{Relationship to the literature}

The study of centrality in networks has evolved to be a significant part of network science \citep{Newman2010book,Barabasi2016}. In economics and the social sciences there has been a focus on Bonacich centrality in social networks. This centrality measure is founded on the eigenvector of the adjacency matrix that represents the network \citep{Bonacich1987}. In economics this has been linked to performance indicators of network representations of economic interactions such as production networks \citep{KeyPlayer2006,Prodnetworks2016,Allouch2021}. The nature of these networks is that they are undirected and, therefore, fundamentally different from the hierarchical networks considered here.

Traditionally, the investigation of directed networks focussed on \emph{degree centrality}---measuring direct dominance relationships \citep{BrinkGilles2003,BrinkRusinowska2022}---and on \emph{betweenness centrality}, which considers the position of nodes in relation to membership of (critical) pathways in the directed network \citep{Bavelas1948,White1994,Newman2003betweenness,Arrigo2018}.

Authority and control in networks has only more recently been investigated from different perspectives. \citet{Barabasi2012} considers an innovative perspective founded on control theory. More prevalent is the study of centrality in hierarchical networks through the $\beta$-measure and its close relatives. \citet{BrinkGilles1994} introduced the $\beta$-measure as a natural measure of influence and considered some non-game theoretic characterisations. The $\beta$-measure is closely related to the PageRank measure introduced by \citet{PageRank1998} and considered throughout the literature on social network centrality measurement.

The $\beta$-measure has been linked to game theoretic measurement of centrality in directed networks by \citet{BrinkGilles2000} and \citet{BrinkBorm2002}. The $\beta$-measure was identified as the Shapley value of the standard successor representations as TU-representations of domination and control in directed networks. \citet{BrinkBorm2005} develop this further through additional characterisations. \citet{Brink2023Network} introduce other, more advanced TU-representations of directed networks and study their Shapley values. They consider a family of centrality measures resulting from this methodology.

\citet{CentralityGameTheory2003}, \citet{CentralityGameTheory2011} and \citet{CentralityGameTheory2018} introduce and explore a game theoretic methodology for measuring network power that is fundamentally different from the methodology used in this paper and the literature reviewed above.  These authors consider a well-chosen TU-game on a networked population of players and subsequently compare the allocated payoffs based on the Shapley value in the unrestricted game with the Shapley value of the network-restricted TU-game. The normalisation of the generated differences now exactly measure the network-positional effects on the players, which can be interpreted as a centrality measure. 

Finally, with regard to the Gately value as a solution concept for TU-games, this conception was seminally introduced for some specific 3-player cost games by \citet{Gately1974}. This contribution inspired a further development of the underlying conception of ``propensity to disrupt'' by \citet{Gately1976} and \citet{Gately1978}, including the definition of several related solution concepts. \citet{Gately1976} also developed an example of a 4-player TU-game in which the Gately value is not a Core imputation. More recently, \citet{Gately2019} generalised the scope of the Gately value and identified exact conditions under which this value is well-defined. This has further been developed by \citet{GillesMallozzi2023}, which showed that the Gately value is always a Core imputation for 3-player games, devised an axiomatisation for the Gately value for arbitrary TU-games, and introduced a generalised Gately value founded on weighted propensities to disrupt. 

\paragraph{Structure of the paper}

Section 2 discusses the foundations of the game theoretic approach that is pursued in this paper. It defines the successor representations and presents their main properties. Furthermore, the standard solution concepts of the Core and the Shapley value are applied to these successor representations. Section 3 introduces the Gately measure, which represents a different philosophy of measuring the exercise of control and power in a network. We investigate when the Gately measure assigns a Core power gauge to a network and we devise an axiomatisation of the Gately measure. The paper concludes with a comprehensive comparison of the Gately and $\beta$-measures, identifying exact conditions under which these two measures are equivalent.

\section{Game theoretic representations of hierarchical networks}

In our study, we focus on networks with directed links, where each link has specifically the interpretation of being the representation of a hierarchical relationship. In a directed network, the direction of a link indicates that one node is positioned as a predecessor while the other node is considered a successor in that particular relationship. Here we interpret this explicitly as a control or authority relationship. Therefore, we denote these networks as \emph{hierarchical} throughout this paper. 

In hierarchical networks, predecessors exercise some form of authority over its successors, allowing the assignment of that control to that particular node. This results in a natural game theoretic representation. We explore these game theoretic representations in this section and investigate the properties of these games.

\paragraph{Notation: Representing hierarchical networks}

Let $N = \{ 1, \ldots , n \}$ be a finite set of nodes, where $n \in \mathbb N$ is the number of nodes considered. Usually we assume that $n \geqslant 3$. A \emph{hierarchical network} on $N$ is a map $D \colon N \to 2^N$ that assigns to every node $i \in N$ a set of \emph{successors} $D(i) \subseteq N \setminus \{ i \}$. We explicitly exclude that a node succeeds itself, i.e., $i \notin D(i)$. The class of all directed networks on node set $N$ is denoted as $\mathbb D^N = \{ D \mid D \colon N \to 2^N$ with $i \notin D(i)$ for all $i \in N \}$.\footnote{We emphasise that in our setting, hierarchical networks are not necessarily tiered or top-down. Hence, we allow these networks to contain cycles and even binary relationships. This allows the incorporation of sports competitions and other social activities to be represented by these hierarchical networks.}

Inversely, in a directed network $D \in \mathbb D^N$, for every node $i \in N$, the subset $D^{-1} (i) = \{ j \in N \mid i \in D(j) \}$ denotes the set of its \emph{predecessors} in $D$. Due to the general nature of the networks considered here, we remark that it might be the case that $D(i) \cap D^{-1} (i) \neq \varnothing$, i.e., some nodes can be successors as well as predecessors of a node.

We introduce the following additional notation to count the number of successors and predecessors of a node in a network $D \in \mathbb D^N \colon$
\begin{numm}
	\item The map $s_D \colon N \to \mathbb N$ counts the number of successors of a node defined by $s_D (i) = \# D(i)$ for $i \in N$;
	\item The map $p_D \colon N \to \mathbb N$ counts the number of predecessors of a node defined by $p_D (i) = \# D^{-1} (i)$ for $i \in N$;
\end{numm}
The previous analysis leads to a natural partitioning of the node set $N$ into different classes based on the number of predecessors of the nodes in a given network $D \in \mathbb D^N \colon$
\begin{align*}
	N^o_D & = \{ i \in N \mid D^{-1} (i) = \varnothing \} = \{ i \in N \mid p_D (i) =0 \} \\
	N_D & = N \setminus N^o_D = \{ i \in N \mid p_D (i) \geqslant 1 \} \\[1ex]
	N^a_D & = \{ i \in N \mid p_D (i) = 1 \} \\
	N^b_D & = \{ i \in N \mid p_D (i) \geqslant 2 \}
\end{align*}
Note that $N = N^o_D \cup N_D$, $N_D = N^a_D \cup N^b_D$. In particular, $\{ N^o_D , N^a_D , N^b_D \}$ forms a partitioning of the node set $N$. We introduce counters $n_D = \# N_D$, $n^a_D = \# N^a_D$ and $n^b_D = \# N^b_D$, leading to the conclusion that $n_D = n^a_D + n^b_D$.

The constructed partitioning informs the analysis of the game theoretic representation of the hierarchical authority structure imposed by $D$ on the node set $N$. Our analysis will show that for certain centrality measures, the class of nodes that have multiple predecessors $N^b_D$ plays a critical role in the determination of the assignment of a power index to these predecessors.

The partitioning of the node set $N$ based on the structure imposed by $D \in \mathbb D^N$ allows further notation to be introduced for every node $i \in N \colon$
\begin{itemize}
	\item $s^a_D (i) = \# \left[ \, D(i) \cap N^a_D \, \right]$ and $s^b_D (i) = \# \left[ \, D(i) \cap N^b_D \, \right]$, resulting in the conclusion that $s_D (i) = s^a_D (i) + s^b_D (i)$.
	\item From the definitions above we conclude immediately that
	\[
	\sum_{i \in N} s^a_D (i) = \# N^a_D
	\]
	and
	\[
	\sum_{i \in N} s^b_D (i) = \sum_{j \in N^b_D} p_D (j) .
	\]
\end{itemize}

\paragraph{Classes of hierarchical networks}

The next definition introduces some normality properties on hierarchical networks that will be used for certain theorems.
\begin{definition}
	Let $D \in \mathbb D^N$ be some hierarchical network on node set $N$.
	\begin{abet}
		\item The network $D$ is \textbf{weakly regular} if for all nodes $i,j \in N^b_D \colon p_D (i) = p_D (j)$. \\ The collection of weakly regular hierarchical networks is denoted by $\mathbb D^N_w \subset \mathbb D^N$.
		\item The network $D$ is \textbf{regular} if for all nodes $i,j \in N_D \colon p_D (i) = p_D (j)$. \\ The collection of regular hierarchical networks is denoted by $\mathbb D^N_r \subset \mathbb D^N_w$.
		\item The network $D$ is \textbf{simple} if for every node $i \in N_D \colon p_D (i) =1$. \\ The collection of simple hierarchical networks is denoted by $\mathbb D^N_s \subset \mathbb D^N_r$.
	\end{abet}
\end{definition}
In a regular network, each node has either no predecessors, or a given fixed number of predecessors. Hence, all nodes with predecessors have exactly the same number of predecessors. In a weakly regular network, each node has either no predecessors, or exactly one predecessors, or a given fixed number $p \geqslant 2$ of predecessors.

The notion of a simple network further strengthens the requirement of a regular network. It imposes that all nodes either have no predecessors, or exactly one predecessor. 

Furthermore, \citet{BrinkBorm2002} introduced the notion of a \emph{simple subnetwork} of a given network $D \in \mathbb D^N$ on the node set $N$. We elaborate here on that definition.
\begin{definition}
	Let $D \in \mathbb D^N$ be a given hierarchical network on $N$. \\
	A network $T \in \mathbb D^N$ is a \textbf{simple subnetwork} of $D$ if it satisfies the following two properties:
	\begin{numm}
		\item For every node $i \in N \colon T(i) \subseteq D(i)$, and
		\item For every node $j \in N_D \colon p_T (j)=1$.
	\end{numm}
	The collection of a simple subnetworks of $D$ is denoted by $\mathcal S(D)$.
\end{definition}
The collection of simple subnetworks of a given network can be used to analyse the Core of the game theoretic representations of hierarchical games as shown below. It is easy to establish that a hierarchical network $D$ is simple if and only if $\mathcal S (D) = \{ D \}$.

\subsection{Game theoretic representations of hierarchical networks}

Using the notation introduced above, we are able to device cooperative game theoretic representations of hierarchical networks. We recall that a \emph{cooperative game with transferable utilities\/}---or a TU-game---on the node set $N$ is a map $v \colon 2^N \to \mathbb R$ such that $v ( \varnothing ) =0$. A TU-game $v$ assigns to every group of nodes $H \subseteq N$ a certain ``worth'' $v(H) \in \mathbb R$. A group of nodes $H \subseteq N$ is also denoted as a \emph{coalition} of nodes, to use a more familiar terminology from cooperative game theory.

To embody the control or authority represented by a hierarchical network $D \in \mathbb D^N$ on the node set $N$ as a cooperative game, we introduce some additional notation. For every group of nodes $H \subseteq N$ we denote
\begin{equation}
	D(H) = \{ j \in N \mid D^{-1} (j) \cap H \neq \varnothing \} = \cup_{i \in H} \, D(i)
\end{equation}
as the \emph{(weak) successors} of coalition $H$ in $D$. A node is a (weak) successor of a node group if at least one of its predecessors is a member of that group.

Similarly, we introduce
\begin{equation}
	D^* (H) = \{ j \in N \mid \varnothing \neq D^{-1} (j) \subseteq H \} = \{ j \in N_D \mid D^{-1} (j) \subseteq H \}
\end{equation}
as the \emph{strong successors} of coalition $H$ in $D$. A node is a strong successor of a node group if \emph{all} predecessors of that node are members of that group. Clearly, strong successors of a node group are completely controlled by the nodes in that particular group and full control can be exercised. This compares to regular or weak successors of a node group over which the nodes in that group only exercise partial control.

The next definition introduces the two main cooperative game theoretic embodiments of this control over other nodes in a network.
\begin{definition} \label{def:SuccessorGames}
	Let $D \in \mathbb D^N$ be some hierarchical network on node set $N$.
	\begin{abet}
		\item The \textbf{successor representation} of $D$ is the TU-game $s_D \colon 2^N \to \mathbb N$ for every coalition $H \subseteq N$ given by $s_D (H) = \# D(H)$, the number of successors of the coalition $H$ in the network $D$. 
		\item We additionally introduce two \textbf{partial successor representations} as the two TU-games $s^a_D \, , \, s^b_D \colon 2^N \to \mathbb N$, which  for every coalition $H \subseteq N$ are given by $s^a_D (H) = \# \left[ \, D(H) \cap N^a_D \, \right]$ and $s^b_D (H) = \# \left[ \, D(H) \cap N^b_D \, \right]$.
		\item The \textbf{strong successor representation} of $D$ is the TU-game $\sigma_D \colon 2^N \to \mathbb N$ for every coalition $H \subseteq N$ given by $\sigma_D (H) = \# D^*(H)$, the number of strong successors of the coalition $H$ in the network $D$.
	\end{abet}
\end{definition}
The successor representation is also known as the ``successor game'' in the literature and the strong successor representation $\sigma_D$ as the ``conservative successor game'' on $D$ \citep{Gilles2010}. It is clear that the four TU-games introduced in Definition \ref{def:SuccessorGames} embody different aspects of the control exercised over nodes in a given hierarchical network. In particular, these TU-games count the number of successors that are under the control of nodes in a selected coalition.

\paragraph{Properties of successor representations}

The next list collects some simple properties of these four games introduced here.
\begin{proposition}
	Let $D \in \mathbb D^N$ be some hierarchical network on node set $N$. Then the following properties hold regarding the successor representations $s_D$, $s^a_D$, $s^b_D$ and $\sigma_D$:
	\begin{numm}
		\item For every node $i \in N \colon s_D ( \, \{ i \} \, ) = s_D (i)$ and the worth of the whole node set is determined as $s_D (N) = n_D = \# N_D$.
		\item $s_D = s^a_D + s^b_D$.
		\item For every coalition $H \subseteq N \colon s^a_D (H) = \sum_{i \in H} s^a_D (i)$, implying that the partial successor representation $s^a_D$ is an additive game.
		\item For every coalition $H \subseteq N \colon s^b_D (H) \leqslant \sum_{i \in H} s^b_D (i)$.
		\item $\sigma_D = s^a_D + \hat\sigma_D$ where for every coalition $H \subseteq N \colon\hat\sigma_D (H) = \sigma_D (H) - s^a_D (H) \leqslant s^b_D (H)$.
		\item For every node $i \in N \colon \sigma_D ( \, \{ i \} \, ) = s^a_D (i)$ and the worth of the whole node set is determined as $\sigma_D (N) = n_D = \# N_D$.
	\end{numm}
\end{proposition}
These properties follow straightforwardly from the definitions, therefore a proof is omitted.

The next theorem collects some properties of the successor representations that have not been remarked explicitly in the literature on cooperative game theoretic approaches to representations of hierarchical networks.\footnote{We recall that the \emph{unanimity game} of coalition $H \neq \varnothing$ is defined by $u_H \colon 2^N \to \{ 0,1 \}$ such that $u_H (T) =1$ if and only if $H \subseteq T \subseteq N$. This implies that $u_H(T) =0$ for all other coalitions $T \subseteq N$.}
\begin{theorem} \label{thm:GameProperties}
	Let $D \in \mathbb D^N$ be some hierarchical network on node set $N$. Then the following properties hold for the successor representations $s_D$ and $\sigma_D \colon$
	\begin{numm}
		\item The strong successor representation $\sigma_D$ is the dual of the successor representation $s_D$ in the sense that
		\begin{equation}
			\sigma_D (H) = s_D(N) - s_D (N \setminus H) \qquad \mbox{for all } H \subseteq N.
		\end{equation}
		\item The strong successor representation is decomposable into unanimity games with
		\begin{equation}
			\sigma_D = \sum_{j \in N_D} u_{D^{-1} (j)}.
		\end{equation}
		\item The strong successor representation $\sigma_D$ is a convex TU-game \citep{Shapley1971} in the sense that $\sigma_D (H) + \sigma_D (K) \leqslant \sigma_D (H \cup K) + \sigma_D (H \cap K)$ for all $H,K \subseteq N$.
		\item The successor representation $s_D$ is concave in the sense that $s_D (H) + s_D (K) \geqslant s_D (H \cup K) + s_D (H \cap K)$ for all $H,K \subseteq N$
	\end{numm}
\end{theorem}
\begin{proof}
	Let $D \in \mathbb D^N$ be some hierarchical network on node set $N$ and let the TU-games $s_D$ and $\sigma_D$ be as defined in Definition \ref{def:SuccessorGames}. \\
	To show assertion (i), let $H \subseteq N$, then it holds that
	\begin{align*}
		s_D (N) - s_D (N \setminus H) & = n_D - \# D(N \setminus H) = n_D - \# \{ j \in N \mid D^{-1}(j) \cap (N \setminus H) \neq \varnothing \} \\
		& = n_D - \# \{ j \in N \mid D^{-1} (j) \setminus H \neq \varnothing \} = \# \{ i \in N_D \mid D^{-1}(i) \setminus H = \varnothing  \} \\
		& = \# \{ i \in N_D \mid D^{-1} (i) \subseteq H \} = \sigma_D (H) .
	\end{align*}
	This shows that $\sigma_D$ is indeed the dual game of $s_D$. \\[1ex]
	Assertion (ii) is Lemma 2.2 in \citet{BrinkBorm2002} and assertion (iii) follows immediately from (ii). Finally, assertion (iv) is implied by the fact that $s_D$ is the dual game of $\sigma_D$---following from assertion (i)---and $\sigma_D$ is convex. 
\end{proof}

\bigskip\noindent
The duality between the successor representation and the strong successor representation implies that some cooperative game theoretic solution concepts result in exactly the same outcomes for both games. In particular, we refer to the Core, the Weber set, the Shapley value, and the Gately value of these successor representations as explored below.

\subsection{Some standard solutions of the successor representations}

The cooperative game theoretic approach to measuring power or hierarchical centrality is based on the assignment of a quantified control gauge to every individual node in a given hierarchical network. A power or hierarchical centrality measure now refers to a rule or procedure that assigns to every node in any hierarchical network such a gauge. In this section we set out the foundations for this approach.
\begin{definition}
	Let $D \in \mathbb D^N$ be some hierarchical network. A \textbf{power gauge} for $D$ is a vector $\delta \in \mathbb R^N_+$ such that $\sum_{i \in N} \delta_i = n_D$. \\
	A \textbf{power measure} on $\mathbb D^N$ is a function $m \colon \mathbb D^N \to \mathbb R^N_+$ such that $\sum_{i \in N} \, m_i (D) = n_D$ for every hierarchical network $D \in \mathbb D^N$.
\end{definition}
The normalisation of a power gauge for a network $D \in \mathbb D^N$ to the allocation of the total number of nodes in $N_D$ is a yardstick that is adopted in the literature, which we use here as well. This normalisation is in some sense arbitrary, but it allows a straightforward application of the cooperative game theoretic methodology as advocated here.

The game theoretic approach adopted here allows us to apply basic solution concepts to impose well-accepted properties on power gauges and power measures. The well-known notion of the Core of a TU-game imposes lower bounds on the power gauges in a given hierarchical network. This leads to the following notion.
\begin{definition}
	A \textbf{Core power gauge} for a given hierarchical network $D \in \mathbb D^N$ is a power gauge  $\delta \in \mathbb R^N_+$ which satisfies that for every group of nodes $H \subseteq N \colon \sum_{j \in H} \delta_j \geqslant \sigma_D (H) = \# D^* (H)$.
	\\
	The set of Core power gauges for $D$ is denoted by $\mathcal C (D) \subset \mathbb R^N_+$.
\end{definition}
The Core requirements on a power gauge impose that every group of nodes is collectively assigned at least the number of nodes that it fully controls. This seems a rather natural requirement. The following insight investigates the structure of the set of Core power gauges for a hierarchical network.
\begin{proposition}
	Let $D \in \mathbb D^N$ be some hierarchical network on node set $N$. Then the following hold:
	\begin{numm}
		\item If $D$ is a simple hierarchical network, then there exists a unique Core power gauge, $\mathcal C (D) = \left\{ \delta^D \right\}$, where $\delta^D_i = s_D (i)$ for every node $i \in N$.
		\item More generally, $\mathcal C (D)$ is equal to the Weber set of $\sigma_D$, which is the convex hull of the unique Core power gauges of all simple subnetworks of $D$ given by $\mathcal C(D) = \mathrm{Conv} \, \left\{ \delta^T \mid T \in \mathcal S(D) \, \right\} \neq \varnothing$.
	\end{numm}
\end{proposition}
\begin{proof}
	Let $D$ be a simple hierarchical network. Hence, $p_D (i) =1$ for all $i \in N_D$. Therefore, for every group of nodes $H \subseteq N$ it holds that $\sigma_D (H) = s_D (H) = \sum_{j \in H} s_D (j)$. Therefore, $\delta^D$ as defined above satisfies the Core requirement for every $H \subseteq N$. \\ Furthermore, suppose that $\delta \in \mathcal C(D)$. Then from $\delta_i \geqslant s_D (i) = \delta^D_i$ for every node $i \in N$ and $\sum_{j \in N} \delta_j = n_D = \sum_{j \in N} \delta^D_j$ it immediately follows that $\delta_i = \delta^D_i$ for all $i \in N$. This shows that $\delta^D$ is the unique Core power gauge for the simple hierarchical network $D$, showing assertion (i).
	\\[1ex]
	Assertion (ii) follows immediately from Theorem 4.2 in \citet{BrinkBorm2002} in combination with assertion (i).
\end{proof}

\paragraph{The $\beta$-measure}

A well-established power measure for hierarchical networks was first introduced by \citet{BrinkGilles1994} and further developed in \citet{BrinkGilles2000} and \citet{BrinkBorm2005}. This $\beta$-\emph{measure} is for every node $i \in N$ defined by
\begin{equation}
	\beta_i (D) = \sum_{j \in D(i)} \frac{1}{p_D (j)} = s^a_D (i) + \sum_{j \in D(i) \cap N^b_D} \frac{1}{p_D (j)}
\end{equation}
The following proposition collects the main insights from the literature on the $\beta$-measure.
\begin{proposition}
	Let $D \in \mathbb D^N$ be a hierarchical network. Then the following properties hold:
	\begin{numm}
		\item $\beta (D) \in \mathcal C(D)$ is the geometric centre of the set of Core power gauges of $D$.
		\item $\beta (D) = \varphi ( s_D ) = \varphi ( \sigma_D)$, where $\varphi$ is the Shapley value\footnote{It is well-established that every TU-game $v \colon 2^N \to \mathbb R^N$ can be written as $v = \sum_{H \subseteq N} \Delta_v (H) \, u_H$, where $\Delta_v (H)$ is the Harsanyi dividend of coalition $H$ in the game $v$ \citep{Harsanyi1959}. Now, the Shapley value is for every $i \in N$ defined by $\varphi_i (v) = \sum_{H \subseteq N \colon i \in H} \tfrac{\Delta_v (H)}{\# \, H}$. Hence, the Shapley value fairly distributes the generated Harsanyi dividends over the players that generate these dividends. The Shapley value was seminally introduced by \citet{Shapley1953}.} on the collection of all cooperative games on $N$.
	\end{numm}
\end{proposition}

\section{The Gately power measure}

Let $v \colon 2^N \to \mathbb R$ be a TU-game on the node set $N$ with $\sum_{i \in N} v( \{ i \} ) \leqslant v(N) \leqslant \sum_{i \in N} M_i (v)$ where $M_i (v) = v(N) - v(N-i)$. Then the \emph{Gately value} of the game $v$ is given by $g(v) \in \mathbb R^N$, which is defined for every node $i \in N$ by
\begin{equation} \label{eq:GatelyValue}
	g_i (v) = v( \{ i \} ) + \frac{M_i (v) - v( \{ i \} )}{\sum_{j \in N} \left( M_j (v) - v( \{ j \} ) \, \right)} \left[ v(N) - \sum_{j \in N} v(\{ j \} ) \, \right]
\end{equation}
The Gately value was seminally introduced by \citet{Gately1974} and further developed by \citet{Gately1976}, \citet{Gately1978}, \citet{Gately2019} and \citet{GillesMallozzi2023}.

We apply the Gately value to the two successor representations formulated above. We show that, similar to the $\beta$-measure, both the regular successor representation and the conservative successor representation result in the same Gately value, defining the \emph{Gately power measure\/}. 
\begin{theorem} \label{prop:GatelyMeasure}
	Let $D \in \mathbb D^N$ be a directed network on node set $N$. Then
	\begin{equation}
		g \left( s_D \right) = g \left( \sigma_D \right) = \xi (D)
	\end{equation}
	where $\xi \colon \mathbb D^N \to \mathbb R^N$ is introduced as the \textbf{Gately power measure} on the class of hierarchical networks $\mathbb D^N$ on $N$ with
	\begin{equation}
	\hspace*{-2em}
		\xi_i (D) = \left\{
		\begin{array}{ll}
			s^a_D (i) + \frac{s^b_D (i)}{\sum_{j \in N^b_D} p_D (j)} \, n^b_D  & \mbox{if } N^b_D \neq \varnothing \\[1ex]
			s^a_D (i) & \mbox{if } N^b_D = \varnothing
		\end{array}
		\right.
	\end{equation}
	for every node $i \in N$. \\
	Furthermore, the Gately measure $\xi$ is the unique power measure that balances the propensities to disrupt a network given by
	\begin{equation}
		\frac{s^b_D (i)}{s_D (i) - \xi_i(D)} = \frac{s^b_D (j)}{s_D (j) - \xi_j(D)}
	\end{equation}
	over all nodes $i,j \in N_D$. 
\end{theorem}
\begin{proof}
	Let $D \in \mathbb D^N$ be such that $N^b_D \neq \varnothing$. Then its successor representation $s_D$ is characterised for every $i \in N$ by
	\begin{align*}
		s_D (N) & = n_D \\
		s_D ( \{ i \} ) & = s_D (i) = \# D(i) \\
		s_D (N-i) & = n_D - \# \left\{ j \in N^a_D \mid D^{-1} (j) = \{ i \} \, \right\} = n_D - s^a_D (i)
	\end{align*}
	From this it follows that $M_i ( s_D ) = s_D(N) - s_D (N-i) = s^a_D (i)$ for every $i \in N$. Since $N^b_D \neq \varnothing$, this implies furthermore that $s_D (i) \geqslant M_i (s_D)$ for every $i \in N$. Therefore, the Gately value can be applied to this game. \\
	From the previous we further derive that
	\[
	s_D (i) - M_i (s_D) = \# \left\{ j \in N \mid \{ i \} \varsubsetneq D^{-1}(j) \, \right\} =s^b_D (i) 
	\]
	and that
	\begin{align*}
		\sum_{j \in N} s_D ( \{ j \} ) - s_D (N) & = \sum_{ j \in N} s_D(j) - n_D = \sum_{h \in N} p_D (h) - n_D \\
		& = \sum_{ j \in N^b_D} \left( \, p_D (j) -1 \right) = \sum_{ j \in N^b_D} p_D(j) - n^b_D .
	\end{align*}
	These properties imply that $s_D$ is a regular TU-game as defined in \citet{GillesMallozzi2023}. This implies that the Gately value applies to $s_D$. \\
	We now compute the Gately value of the successor representation $s_D$. We note here that $s_D$ is a concave cost game, implying that the reverse formulation of (\ref{eq:GatelyValue}) needs to be applied. Hence, we derive for every $i \in N$ that
	\begin{align*}
		g_i ( s_D ) & = s_D ( \{ i \} ) - \frac{ s_D ( \{ i \} ) - M_i (s_D)}{\sum_{j \in N} \left(  s_D ( \{ i \} ) - M_j (s_D) \right)} \cdot\left( \sum_{j \in N} s_D ( \{ j \} ) - s_D (N) \right) \\[1ex]
		& = s_D (i) - \frac{s^b_D (i)}{\sum_{j \in N} s^b_D (j)} \cdot\left( \sum_{ j \in N^b_D} p_D(j) - n^b_D \, \right) \\[1ex]
		&= s_D (i) - \frac{s^b_D (i)}{\sum_{j \in N^b_D} p_D (j)} \cdot\left( \sum_{ j \in N^b_D} p_D (j) - n^b_D \, \right) \\[1ex]
		& = s_D (i) - s^b_D (i) + \frac{s^b_D (i)}{\sum_{j \in N^b_D} p_D (j)} \cdot n^b_D  \\[1ex]
		& = s^a_D (i)+ \frac{s^b_D (i)}{\sum_{j \in N^b_D} p_D (j)} \cdot n^b_D = \xi_i (D)
	\end{align*}
	Similarly, the conservative successor representation $\sigma_D$ for the hierarchical network $D$ is characterised for every $i \in N$ by
	\begin{align*}
		\sigma_D (N) & = n_D \\
		\sigma_D (i) & = s^a_D (i) \\
		\sigma_D (N-i) & = n_D - s_D (i)
	\end{align*}
	For the conservative successor representation $\sigma_D$ we derive from the above that $M_i (\sigma_D) = s_D (i)$, implying that $\sigma_D (i) \leqslant M_i ( \sigma_D)$ for every $i \in N$. Also, $\sigma_D (i) < M_i ( \sigma_D)$ for some $i \in N$, since $N^b_D \neq \varnothing$. Therefore, $\sigma_D$ is regular as defined in \citet{GillesMallozzi2023}. Hence, the Gately value applies to $\sigma_D$. \\ Since $\sigma_D$ is a convex game, the formulation stated in (\ref{eq:GatelyValue}) applies. Now, we compute that
	\[
	M_i (\sigma_D) - \sigma_D (i) = s_D (i) - s^a_D (i) = s^b_D (i)
	\]
	and
	\begin{align*}
		g_i ( \sigma_D ) & = \sigma_D ( \{ i \} ) + \frac{M_i (\sigma_D \, ) - \sigma_D ( \{ i \} )}{\sum_{j \in N} \left( M_j (\sigma_D \, ) - \sigma_D ( \{ j \} ) \right)} \cdot\left( \sigma_D (N) - \sum_{j \in N} \sigma_D ( \{ j \} ) \right) \\[1ex]
		& = s^a_D (i) + \frac{s^b_D (i)}{\sum_{j \in N} s^b_D (j)} \cdot n^b_D = s^a_D (i) + \frac{s^b_D (i)}{\sum_{j \in N^b_D} p_D (j)} \cdot n^b_D = \xi_i (D)
	\end{align*}
	This shows the first equality in the assertion of the proposition. \\[1ex]
	Next, let $D \in \mathbb D^N$ be such that $N^b_D = \varnothing$. Then $p_D (j) =1$ for all $j \in N_D$. This implies that for every $i \in N \colon M_i ( s_D ) = s_D ( \{ i \} ) = s_D (i)$. Hence, for $i \in N \colon$
	\[
	g_i (s_D) = s_D( \{ i \} ) =  s_D (i) = s^a_D (i) =  \xi_i (D) .
	\]
	Finally, for every $i \in N \colon M_i ( \sigma_D ) = s_D (i) = s^b_D (i) = 0$. Hence,
	\[
	g_i (\sigma_D) = \sigma_D ( \{ i \} ) =  s^a_D (i) = \xi_i (D) .
	\]
	Combined with the previous case, this shows the first assertion of the proposition. \\[1ex]
	Finally, the second assertion of the proposition follows immediately from identifying the propensity to disrupt \citep[Definition 3.2]{GillesMallozzi2023} in the successor representation $s_D$ for some game theoretic imputation $m \in \mathbb R^N$ as
	\[
	\frac{M_i (s_D) - s_D ( \{ i \} )}{m_i - s_D ( \{ i \} )} = \frac{s^b_D (i) - s_D (i)}{m_i - s_D (i)} = \frac{s^b_D (i)}{m_i - s_D (i)} .
	\]
	Using the definition of a Gately point \citep[Definition 3.2]{GillesMallozzi2023} and noting that $\xi_i (D) = s_D (i) =0$ for every $i \in N^o_D$, the second assertion of the theorem is confirmed.
\end{proof}

\bigskip\noindent
The Gately power measure introduced in Theorem \ref{prop:GatelyMeasure} is founded on fundamentally different principles than the $\beta$-measure or other power measures. Now, the Gately power measure is member of a family of values that considers the control exercise over the nodes in $N^b_D$ to be a collective resource in any hierarchical network $D \in \mathbb D^N$. The control is then allocated according to some well-chosen principle. In particular, The Gately measure allocates the control over $N^b_D$ \emph{proportionally} to the predecessor of the nodes in $N^b_D$. Assuming $N^b_D \neq \varnothing$, this proportional allocator is for every node $i \in N$ with $D(i) \cap N^b_D \neq \varnothing$ defined as
\begin{equation}
	a_i (D) = \frac{s^b_D (i)}{\sum_{j \in N} s^b_D (j)} = \frac{s^b_D (i)}{\sum_{h \in N^b_D} p_D (h)}
\end{equation}
where the Gately measure is now given by $\xi_i (D) = s^a_D (i) + a_i (D) \cdot n^b_D$.

This compares, for example, to the allocation principle based on the egalitarian allocator of the power over the nodes in $N^b_D \neq \varnothing$ given by
\begin{equation}
	e_i (D) = \frac{1}{\# \, \{ j \in N \mid D(i) \cap N^b_D \neq \varnothing \}}
\end{equation}
and the resulting \emph{Restricted Egalitarian} power measure given by $\varepsilon_i (D) = s^a_D (i) + e_i (D) \, n^b_D$. We emphasise that the Gately and Restricted Egalitarian power measures are members of the same family of power measures for hierarchical networks, which have a collective allocative perspective on the control over the nodes in $N^b_D$.

\subsection{Properties of the Gately measure}

We investigate the properties of the Gately measure from the cooperative game theoretic perspective developed in this paper. We first investigate whether the Gately measure assigns a Core power gauge as is the case for the $\beta$-measure. Second, we consider some characterisations of the Gately measure. In particular, we derive an axiomatisation as well as investigate some interesting properties of the Gately measure on some special subclasses of hierarchical networks. 

\paragraph{The Gately measure is not necessarily a Core power gauge}

We first establish that contrary to the property that the $\beta$-measure is the geometric centre of the set of Core power gauges of a given hierarchical network, its Gately power gauge does not necessarily have to satisfy the Core constraints. The next example provides a hierarchical network on a node set of 8 nodes which Gately measure is not a Core power gauge.

\bigskip
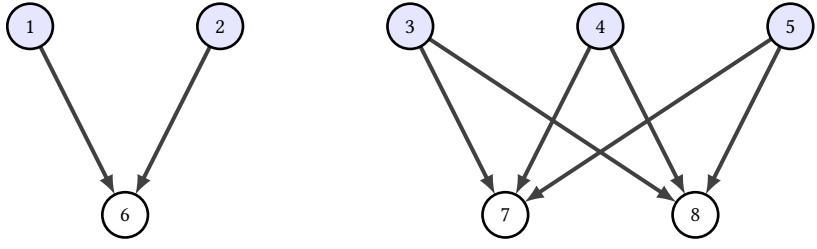
\begin{figure}[h]
\begin{center}
\begin{tikzpicture}[scale=0.25]
\Vertex[x=0,y=10,color=blue!10,label=1]{1}
\Vertex[x=10,y=10,color=blue!10,label=2]{2}
\Vertex[x=20,y=10,color=blue!10,label=3]{3}
\Vertex[x=30,y=10,color=blue!10,label=4]{4}
\Vertex[x=40,y=10,color=blue!10,label=5]{5}
\Vertex[x=5,y=0,color=white,label=6]{6}
\Vertex[x=25,y=0,color=white,label=7]{7}
\Vertex[x=35,y=0,color=white,label=8]{8}

\Edge[Direct](1)(6)
\Edge[Direct](2)(6)
\Edge[Direct](3)(7)
\Edge[Direct](3)(8)
\Edge[Direct](4)(7)
\Edge[Direct](4)(8)
\Edge[Direct](5)(7)
\Edge[Direct](5)(8)

\end{tikzpicture}
\end{center}
\caption{The hierarchical network considered in Example \ref{xpl:CounterCore}.}
\end{figure}

\begin{example} \label{xpl:CounterCore}
	Consider the hierarchical network $D$ depicted in Figure 1 based on a node set $N = \{ 1, \ldots ,8 \}$. We note that $N_D = N^b_D = \{ 6,7,8 \}$ and that nodes 1 and 2 fully control node 6, while nodes 3, 4 and 5 fully control nodes 7 and 8. \\ We compute that $s^a_D (i) =0$ for all $i \in N$ and that, therefore, any power measure only considers the control arrangements of the nodes in $N_D$. In particular,
	\[
	\beta (D) = \left( \, \tfrac{1}{2}, \tfrac{1}{2} , \tfrac{2}{3} , \tfrac{2}{3} , \tfrac{2}{3} , 0,0,0 \, \right) \in \mathcal C (D)
	\]
	and that
	\[
	\xi (D) = \left( \, \tfrac{3}{8}, \tfrac{3}{8} , \tfrac{3}{4} , \tfrac{3}{4} , \tfrac{3}{4} , 0,0,0 \, \right) \notin \mathcal C (D) .
	\]
	Therefore, $\xi_1 (D) + \xi_2 (D) = \tfrac{3}{4} < 1 = \sigma_D ( \{ 1,2 \} )$ shows that the Gately measure does not allocate sufficient power to the first two agents. The underlying reason is that the Gately value considers the control of \emph{all} nodes in $N^b_D$ to be a collective resource that is proportionally distributed according to $s^b_D (i)$. The relative low values of $s^b_D (1) = s^b_D (2) =1$ in comparison with $s^b_D (3) = s^b_D (4) =s^b_D (5) =2$ imply that the assigned share of the first two nodes is less than the total of fully controlled nodes by that node pair.
\end{example}

The following theorem establishes under which conditions the Gately measure assigns a Core power gauge to a hierarchical network. The class of networks identified in (ii) compares with the classes of networks for which the Gately measure is identical to the $\beta$-measure.
\begin{theorem} \label{thm:GatelyCore}
	Let $D \in \mathbb D^N$ be a hierarchical network on $N$. 
	\begin{numm}
		\item If $\# \, \{ i \in N \mid D(i) \neq \varnothing \} \leqslant 3$, then $\xi (D) \in \mathcal C (D)$.
		\item If $D$ is weakly regular, i.e., for all $i,j \in N^b_D \colon p_D (i) = p_D (j)$, then $\xi (D) \in \mathcal C(D)$.
	\end{numm}
\end{theorem}
\begin{proof}
	Assertion (i) follows directly from Theorem 4.2 of \citet{GillesMallozzi2023}. Here we note that the strong successor representation $\sigma_D$ is essentially a three-player if $\# \, \{ i \in N \mid D(i) \neq \varnothing \} = 3$ and a two-player game if $\# \, \{ i \in N \mid D(i) \neq \varnothing \} =2$. Both cases are covered by Theorem 4.2 of Gilles and Mallozzi, establishing that $\xi (D) = g(\sigma_D) \in \mathcal C ( \sigma_D) = \mathcal C(D)$ as desired. \\[1ex]
	Let $p \geqslant 2$ such that $p_D (j) = p$ for all $j \in N^b_D$ as assumed. From this it follows that for $i \in N \colon$
	\[
	\xi_i (D) = s^a_D (i) + \frac{n^b_D}{\sum_{j \in N^b_D} p_D (j)} \, s^b_D (i) = s^a_D (i) + \frac{n^b_D}{p \cdot n^b_D} \, s^b_D (i) = s^a_D (i) + \tfrac{1}{p} \, s^b_D (i)
	\]
	Now let $H \subseteq N$ be some node group. Define
	\[
	K_H = \left\{ j \in N^b_D \mid D^{-1} (j) \subseteq H \, \right\} \qquad \mbox{and} \qquad k_H = \# K_H \leqslant n^b_D .
	\]
	Now note that
	\[
	\sigma_d (H) = \# \, \left\{ \, j \in N \mid D^{-1} (j) \subseteq H \, \right\} = \sum_{i \in H} s^a_D (i) + k_H .
	\]
	We next compute
	\begin{align*}
		\sum_{i \in H} \xi_i (D) & = \sum_{i \in H} s^a_D (i) + \tfrac{1}{p} \sum_{i \in H} s^b_D (i) \\[1ex]
		& = \sum_{i \in H} s^a_D (i) + \tfrac{1}{p} \sum_{i \in H} \left[ \, \# \left\{ j \in K_H \mid j \in D(i) \, \right\} + \# \left\{ j \in N^b_D \setminus K_H \mid j \in D(i) \right\} \, \right] \\[1ex] 
		& \geqslant  \sum_{i \in H} s^a_D (i) + \tfrac{1}{p} \sum_{i \in H} \, \# \left\{ j \in K_H \mid j \in D(i) \, \right\}  \\[1ex]
		& = \sum_{i \in H} s^a_D (i) + \tfrac{1}{p} \sum_{j \in K_H} p_D (j) = \sum_{i \in H} s^a_D (i) + \tfrac{1}{p} \cdot p \, K_H \\[1ex]
		& = \sum_{i \in H} s^a_D (i) + K_H = \sigma_D (H) .
	\end{align*}
	Since $H$ was arbitrary, this establishes that $\sum_{i \in H} \xi_i(D) \geqslant \sigma_D (H)$ for \emph{all} coalitions $H \subseteq N$ and, therefore, $\xi (D) \in \mathcal C ( \sigma_D ) = \mathcal C (D)$, showing the second assertion of the theorem.
\end{proof}

\paragraph{An axiomatic characterisation of the Gately measure}

We are able to devise a full axiomatisation of the Gately measure on $\mathbb D^N$ based on three defining properties. In order to state these properties, we define for every hierarchical network $D \in \mathbb D^N$ on node set $N$ its \emph{principal restriction} as the network $P_D \in \mathbb D^N$ defined by $P_D (i) = N^b_D \cap D(i)$ for every node $i \in N$. Similarly, a hierarchical network $D \in \mathbb D^N$ is a \emph{principal network} if $D = P_D$. It is clear that a principal hierarchical network $D$ is characterised by the property that $N^a_D = \varnothing$, meaning that all nodes with predecessors have actually multiple predecessors.

\begin{theorem} \label{thm:AxiomatisationA}
	Let $\mathbb D^N$ be the class of hierarchical networks on node set $N$. Then the Gately measure $\xi \colon \mathbb D^N \to \mathbb R^N$ is the unique function $m \colon \mathbb D^N \to \mathbb R^N$ that satisfies the following three properties:
	\begin{numm}
		\item \textbf{Normalisation:} $m$ is $n_D$-normalised in the sense that $\sum_N m_i (D) = n_D$ for all $D \in \mathbb D^N$;
		\item \textbf{Normality:} For every hierarchical network $D \in \mathbb D^N$ it holds that
		\begin{equation}
			m(D) = s^a_D + m (P_D)
		\end{equation}
		where $P_D \in \mathbb D^N$ is the principal restriction of $D$, and;
		\item \textbf{Restricted proportionality:} For every principal network $D \in \mathbb D^N$ with $D=P_D$ it holds that
		\begin{equation}
			m(D) = \lambda_D \, s_D \qquad \mbox{for some } \lambda_D >0 .
		\end{equation}
	\end{numm}
\end{theorem}
\begin{proof}
	We first show that the Gately measure $\xi$ indeed satisfies these three properties. Normalisation trivially follows from the definition of the Gately measure. Let $D \in \mathbb D^N$ be an arbitrary hierarchical network on $N$. \\ First, if $N^b_D = \varnothing$, we have that $s^b_D =0 \in \mathbb R^N$ and $\xi (D) = s^a_D \in \mathbb R^N_+$. Hence, $P_D$ is the empty network with $P_D (i) = \varnothing$ for all $i \in N$. Therefore, $s_{P_D} = s^b_D =0$ and $\xi (P_D)=0 = \lambda \, s^b_D = \lambda \, s_{P_D}$ for any $\lambda >0$. This implies that $\xi (D)$ satisfies the normality property as well as restricted proportionality for the case that $N^b_D = \varnothing$. \\ Second, in the case that $N^b_D \neq \varnothing$, we have that
	\[
	\xi (D) = s^a_D + \frac{n^b_D}{\sum_{j \in N^b_D} p_D (j)} \, s^b_D
	\]
	Furthermore, we compute that for every $i \in N \colon \xi_i(P_D) = \frac{n^b_D}{\sum_{j \in N^b_D} p_D (j)} \, s^b_D (i) = \lambda_D \, s^b_D (i)$, where $\lambda_D = \frac{n^b_D}{\sum_{j \in N^b_D} p_D (j)}$. This implies that $g$ satisfies restricted proportionality. \\[1ex] Next, let $m \colon \mathbb D^N \to \mathbb R^N$ be a power measure that satisfies the three given properties. \\ 
	First, consider a network $D \in \mathbb D^N$ with $N^b_D = \varnothing$. Hence, $s^b_D=0$, implying with (iii) that $m(P_D) = \lambda_D s^b_D =0$. Therefore, with (ii), it follows that $m(D) = s^a_D = \xi (D)$. \\ 
	Next, consider $D \in \mathbb D^N$ with $N^b_D \neq \varnothing$. Then, noting that $s_{P_D} = s^b_D$, from (iii) it follows that $m(P_D) = \lambda_D \, s^b_D > 0$\footnote{We use the definition that for $x,y \in \mathbb R^N \colon x > y$ if and only if $x_i \geqslant y_i$ for all $i \in N$ and $x_j > y_j$ for some $j \in N$.} and with (ii) this implies that
	\[
	m(D) = s^a_D + m (P_D) = s^a_D + \lambda_D \, s^b_D >0.
	\]
	Using the normalisation of $m$ stated as property (i), we conclude that $\lambda_D = \frac{n^b_D}{\sum_{j \in N^b_D} p_D (j)}$ and, therefore, $m(D) = \xi (D)$. \\[1ex]
	Uniqueness of $g$ as a power measure that satisfies the three listed properties is immediate and stated here without proof.
\end{proof}

\bigskip\noindent
The three properties stated in Theorem \ref{thm:AxiomatisationA} have a natural and direct interpretation. In particular, the normality property imposes that the power measure always assigns its uniquely subordinated nodes are assigned to a given node and that the main task of a power measure is to assign a power gauge for the principal restriction of any hierarchical network. This seems a rather natural hypothesis that is satisfied by other power measures such as the $\beta$-measure.

Restricted proportionality imposes that in a principal network the assigned power gauge is proportional to the number of other nodes that are controlled by that node. Again this seems a plausible hypothesis, even though it is violated by the $\beta$-measure.

In fact, the three properties are non-redundant as the following simple examples show:
\begin{itemize}
	\item As indicated above, with regard to the axiomatisation devised in in Theorem \ref{thm:AxiomatisationA}, the $\beta$-measure satisfies the normalisation property (i) as well as the normality property (ii), but not the restricted proportionality property (iii). The Restricted Egalitarian power measure $e$ is another example of a power measure on $\mathbb D^N$ that satisfies (I) as well as (ii), but not the Restricted Proportionality property (iii).
	\item Consider the proportional power measure $\rho$ on $\mathbb D^N$ with for every $D \in \mathbb D^N \colon$
	\[
	\rho (D) = \frac{n_D}{\sum_{i \in N} s_D (i)} \, s_D.
	\]
	Then this proportional power measure satisfies the normalisation property (i) as well as the restricted proportionality property (iii), but not the normality property (ii) stated in Theorem \ref{thm:AxiomatisationA}.
	\item Finally, consider the direct power measure $s$ on $\mathbb D^N$ with for every $D \in \mathbb D^N \colon$
	\[
	s (D) = s_D \in \mathbb R^N_+ .
	\]
	This direct power measure $s$ satisfies the restricted proportionality property (iiI) as well as the normality property (ii) stated in Theorem \ref{thm:AxiomatisationA}, but not the stated normalisation property (i).
\end{itemize}

\subsection{A comparison between the $\beta$-measure and the Gately measure}

On the class of weakly regular hierarchical networks, the Gately value satisfies the strong property that it coincides with the $\beta$-measure. This is explored in the next theorem.
\begin{theorem} \label{thm:GatelyBeta}
	Let $D \in \mathbb D^N_w$ be a weakly regular hierarchical network on $N$. Then the Gately measure coincides with the $\beta$-measure, i.e., $\xi (D) = \beta (D)$.
\end{theorem}
\begin{proof}
	To show the assertion, denote by $p = p_D (i) \geqslant 2$ the common number of predecessors for $i \in N^b_D$. Then it holds that
	\[
	\sum_{j \in N} s^b_D (j) = \sum_{i \in N^b_D} p_D (i) = p \cdot n^b_D .
	\]
	This implies simply that
	\[
	\xi (D) = s^a_D + \frac{n^b_D}{p \cdot n^b_D} \, s^b_D = s^a_D + \tfrac{1}{p} \, s^b_D = \beta (D) ,
	\]
	since for every $i \in N \colon \beta_i (D) = \sum_{j \in D(i)} \tfrac{1}{p_D (j)} = s^a_D (i) + s^b_D (i) \cdot \tfrac{1}{p}$.
\end{proof}

\bigskip\noindent
Theorems \ref{thm:GatelyCore} and \ref{thm:GatelyBeta} allow us to delineate networks with non-empty sets of Core power gauges that contain either the Gately measure, or the $\beta$-measure, or both, as well as determine when both of these measures coincide. This is explored in the next example.

\begin{figure}[h]
\begin{center}
\begin{tikzpicture}[scale=0.5]
\Vertex[x=10,y=10,color=white,label=1]{1}
\Vertex[x=5,y=5,color=blue!10,label=2]{2}
\Vertex[x=15,y=5,color=blue!10,label=3]{3}
\Vertex[x=5,y=0,color=blue!30,label=4]{4}
\Vertex[x=15,y=0,color=blue!30,label=5]{5}

\Edge[Direct](1)(2)
\Edge[Direct](1)(3)
\Edge[Direct](1)(4)
\Edge[Direct](1)(5)
\Edge[Direct](2)(4)
\Edge[Direct](2)(5)
\Edge[Direct](3)(5)
\end{tikzpicture}
\end{center}
\caption{Network for Example \ref{ex:contra}.}
\end{figure}
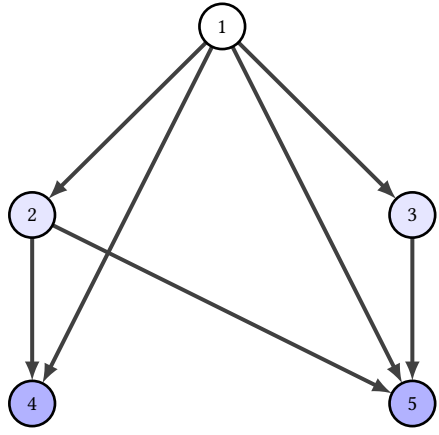

\begin{example} \label{ex:contra}
	Consider the network $D$ depicted in Figure 2 on the node set $N = \{ 1, \ldots , 5 \}$. Note that this network satisfies the conditions of Theorem \ref{thm:GatelyCore}(i), but not of Theorem \ref{thm:GatelyBeta}. Hence, $\xi (D) \in \mathcal C(D) \neq \varnothing$, but $\xi (D) \neq \beta (D) \in \mathcal C(D)$. \\ 
	We compute that the set of Core power gauges for this network is given by\footnote{We point out that there are two simple subnetworks of $D$ that result in exactly the same power gauge, namely the subnetwork in which node 1 dominates node 4 and node 2 dominates node 5 and vice versa. The resulting power gauge is $(3,1,0,0,0)$.}
	\[
	\mathcal C (D) = \mathrm{Conv} \, \left\{ \, (2,1,1,0,0) \, , \, (3,0,1,0,0) \, , \, (3,1,0,0,0) \, , \, (2,2,0,0,0) \, , \, (4,0,0,0,0) \, \right\}
	\]
	Next, we determine that the $\beta$-measure is in the (weighted) centre of $\mathcal C(D)$ with
	\[
	\beta (D) = \left( \, 2 \tfrac{5}{6} \, , \, \tfrac{5}{6} \, , \, \tfrac{1}{3} \, , \, 0 \, , \, 0 \, \right) \in \mathcal C (D)
	\]
	and that the Gately measure is computed as
	\[
	\xi (D) = \left( \, 2 \tfrac{4}{5} \, , \, \tfrac{4}{5} \, , \, \tfrac{2}{5} \, , \, 0 \, , \, 0 \, \right) \in \mathcal C (D) .
	\]
	Clearly, we have established that in this case $\xi (D) \neq \beta (D)$ even though the network $D$ satisfies the condition of Theorem \ref{thm:GatelyCore}(i), implying that the Gately measure is a Core power gauge.
\end{example}

\noindent
A question remaining is whether the assertion of Theorem \ref{thm:GatelyBeta} can be reversed, i.e., if $\xi (D) = \beta (D)$ then $D$ has to be weakly regular. The answer to that is negative as the following example shows.

\begin{figure}[h]
\begin{center}
\begin{tikzpicture}[scale=0.5]
\Vertex[x=0,y=5,color=white,label=1]{1}
\Vertex[x=6,y=5,color=white,label=2]{2}
\Vertex[x=12,y=5,color=blue!10,label=3]{3}
\Vertex[x=16,y=5,color=blue!10,label=4]{4}
\Vertex[x=3,y=0,color=red!10,label=5]{5}
\Vertex[x=10,y=0,color=red!30,label=6]{6}
\Vertex[x=14,y=0,color=red!30,label=7]{7}
\Vertex[x=18,y=0,color=red!30,label=8]{8}

\Edge[Direct](1)(5)
\Edge[Direct](1)(5)
\Edge[Direct](2)(5)
\Edge[Direct](2)(5)
\Edge[Direct](3)(6)
\Edge[Direct](3)(7)
\Edge[Direct](3)(8)
\Edge[Direct](4)(6)
\Edge[Direct](4)(7)
\Edge[Direct](4)(8)
\end{tikzpicture}
\end{center}
\caption{Network for Example \ref{ex:contraMore}.}
\end{figure}
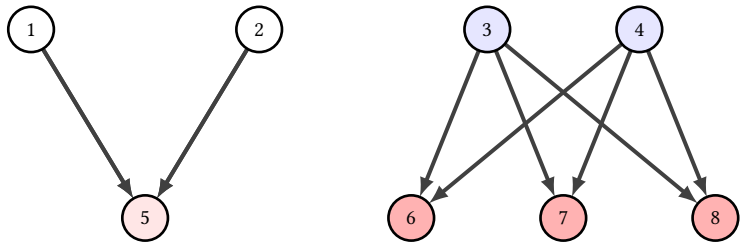

\begin{example} \label{ex:contraMore}
	Consider the node set $N = \{ 1, \ldots , 8 \}$ and the network $D$ depicted in Figure 3. As the colour code indicates, there are four groups of nodes in this network. Nodes 1 and 2 together dominate node 5, while nodes 3 and 4 together dominate nodes 6,7 and 8. \\
	Clearly, this network is not weakly regular. On the other hand, we claim that $\xi (D) = \beta (D)$. Now, we compute that $N^a_D = \varnothing$, $N^b_D = \{ 5,6,7,8 \}$ and $n^b_D =4$. Furthermore, $\sum_{i \in N^b_D} p_D (i) =2+ 3 \cdot 2 =8$. This implies now that for every $i \in N \colon$
	\[
	\xi_i (D) = \frac{n^b_D}{\sum_{i \in N^b_D} p_D (i)} \, s^b_D (i) = \tfrac{4}{8} \, s^b_D (i) = \tfrac{1}{2} \, s^b_D (i)
	\]
	resulting in $\xi (D) = \left( \tfrac{1}{2} \, , \, \tfrac{1}{2} \, , 1\tfrac{1}{2} \, , 1\tfrac{1}{2} \, , \, 0 \, , \, 0 \, , \, 0 \, \right)$ and that this coincides with $\beta (D)$.
\end{example}

\singlespace
\bibliographystyle{ecta}
\bibliography{RPDB}

\end{document}